\documentclass[conference]{IEEEtran}
\usepackage[T1]{fontenc}

\def\argmax{\mathop{\rm \arg\!\max}}

\usepackage{graphicx}
\usepackage[noadjust]{cite}
\usepackage{mcite}
\usepackage{amsfonts,helvet}
\usepackage{fancyhdr}
\usepackage{threeparttable}
\usepackage{epsf,epsfig}
\usepackage{amsthm}
\usepackage{amsmath}
\usepackage{siunitx}
\usepackage{amssymb}
\usepackage{dsfont}
\usepackage{subfigure}
\usepackage{color}
\usepackage{enumerate}
\usepackage{cancel}
\usepackage{bbm}
\usepackage{dsfont}
\usepackage[subnum]{cases}
\usepackage{adjustbox}
\usepackage[linesnumbered,ruled,noend]{algorithm2e}
\usepackage{multicol}
\usepackage[english]{babel}
\usepackage{enumitem}

\newtheorem{theorem}{Theorem}

\newtheorem{corollary}{Corollary}

\newmuskip\pFqmuskip
\def\ba{{\bf a}}

\def\bh{{\bf h}}

\def\bn{{\bf n}}

\def\bq{{\bf q}}
\def\br{{\bf r}}
\def\bs{{\bf s}}

\def\bu{{\bf u}}

\def\bx{{\bf x}}
\def\by{{\bf y}}

\def\bA{{\bf A}}
\def\bB{{\bf B}}
\def\bC{{\bf C}}
\def\bD{{\bf D}}

\def\bG{{\bf G}}
\def\bH{{\bf H}}
\def\bI{{\bf I}}

\def\bP{{\bf P}}
\def\bQ{{\bf Q}}
\def\bR{{\bf R}}

\def\bU{{\bf U}}

\def\bW{{\bf W}}

\def\cA{\mbox{$\mathcal{A}$}}

\def\cC{\mbox{$\mathcal{C}$}}

\def\cN{\mbox{$\mathcal{N}$}}

\def\cP{\mbox{$\mathcal{P}$}}
\def\cQ{\mbox{$\mathcal{Q}$}}

\def\cV{\mbox{$\mathcal{V}$}}

\def\bbC{\mbox{$\mathbb{C}$}}

\def\bbE{\mbox{$\mathbb{E}$}}

\def\bbV{\mbox{$\mathbb{V}$}}

\usepackage{array}

\makeatletter
\newcommand{\thickhline}{%
    \noalign {\ifnum 0=`}\fi \hrule height 1pt
    \futurelet \reserved@a \@xhline
}
\newcolumntype{"}{@{\hskip\tabcolsep\vrule width 1pt\hskip\tabcolsep}}
\makeatother

\IEEEoverridecommandlockouts

\title{
A Hybrid Beamforming Receiver with Two-Stage Analog Combining and Low-Resolution ADCs
}
\author{
Jinseok Choi, $^\dagger$Gilwon Lee, and Brian L. Evans \\
Wireless Networking and Communication Group, The University of Texas at Austin\\
E-mail: jinseokchoi89@utexas.edu, bevans@ece.utexas.edu\\
$^\dagger$Intel Corporation, Santa Clara, CA. E-mail: gilwon.lee@intel.com

\thanks{
The authors at The University of Texas at Austin were supported by gift funding from Huawei Technologies.
}
}

\begin{document}
\maketitle

\begin{abstract}
In this paper, we propose a two-stage analog combining architecture for millimeter wave (mmWave) communications with hybrid analog/digital beamforming and low-resolution analog-to-digital converters (ADCs).
We first derive a two-stage combining solution by solving a mutual information (MI) maximization problem without a constant modulus constraint on analog combiners.
With the derived solution, the proposed receiver architecture splits the analog combining into a channel gain aggregation stage followed by a spreading stage to maximize the MI by effectively managing quantization error.
We show that the derived two-stage combiner achieves the optimal scaling law with respect to the number of radio frequency (RF) chains and maximizes the MI for homogeneous singular values of a MIMO channel.
Then, we develop a two-stage analog combining algorithm to implement the derived solution under a constant modulus constraint for mmWave channels. 
Simulation results validate the algorithm performance in terms of MI.
\end{abstract}


\section{Introduction}
\label{sec:intro}

Millimeter wave communications have attracted large research interest as a promising 5G technology \cite{pi2011introduction, rappaport2013millimeter}.
Utilizing multi-gigahertz bandwidth can potentially achieve an order of magnitude increase in achievable rate \cite{andrews2014will}.
Significant power consumption at receivers with large antenna arrays, however, is considered as one of the primary challenges to address. 
In this paper, we consider hybrid beamforming receivers equipped with low-resolution ADCs to resolve such a challenge by reducing both the number of RF chains and ADC bits.
For mmWave channels, hybrid beamforming techniques were developed by leveraging the sparsity of the channel in the beamspace \cite{el2014spatially, alkhateeb2014channel, bogale2014beamforming, rusu2015low, chen2015iterative, liang2014low, alkhateeb2015limited,mendez2016hybrid}.
Orthogonal matching pursuit (OMP)-based algorithms were proposed in \cite{el2014spatially, alkhateeb2014channel, bogale2014beamforming, rusu2015low, chen2015iterative}  to design analog beamformer by using array response vectors (ARVs).
The OMP-based algorithm proposed in \cite{el2014spatially} was improved by iteratively updating the phases of the phase shifters \cite{chen2015iterative} and by combining OMP and local search to reduce the computational complexity \cite{rusu2015low}.
A channel estimation technique was also developed by using hierarchical multi-resolution codebook-based ARVs for hybrid systems \cite{alkhateeb2014channel}.

Unlike the previous work \cite{ el2014spatially, alkhateeb2014channel, bogale2014beamforming, rusu2015low, chen2015iterative, liang2014low, alkhateeb2015limited,mendez2016hybrid}, hybrid beamforming systems with low-resolution ADCs were investigated in \cite{venkateswaran2010analog,choi2017resolution, choi2018user, mo2017hybrid,  abbas2017millimeter, roth2018comparison}. 
In \cite{venkateswaran2010analog}, an analog combiner was designed by minimizing the mean squared error (MSE) including the quantization error without a constant modulus constraint.
In \cite{mo2017hybrid, abbas2017millimeter}, a singular value decomposition (SVD)-based analog combiner was implemented by using an alternating projection method.
It was shown in \cite{mo2017hybrid} that hybrid MIMO systems with low-resolution ADCs provide the superior performance and power tradeoff compared to infinite-resolution ADC systems. 
In addition, a subarray antenna structure with low-resolution ADCs was investigated in \cite{roth2018comparison}.
The analysis in \cite{mo2017hybrid, abbas2017millimeter, roth2018comparison} provided insights for the hybrid architecture with low-resolution ADCs.
The coarse quantization effect, however, was not explicitly taken into account in the analog combiner design.

%

\begin{figure}[!t]\centering
	\includegraphics[scale = 0.48]{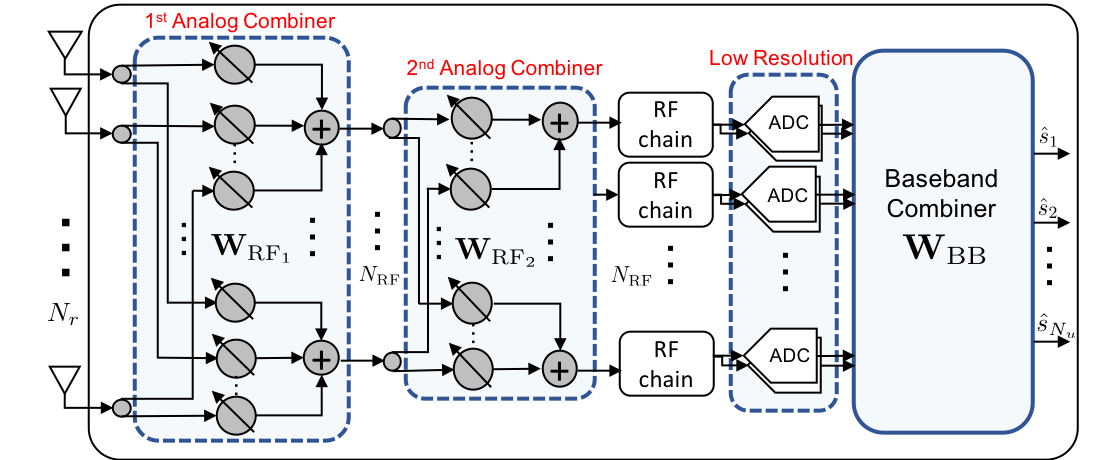}
	\caption{A receiver architecture with two-stage analog combining, low-resolution ADCs and digital combining.} 
	\label{fig:receiver}
\end{figure}

In this paper, we propose a two-stage analog combining architecture for hybrid systems with low-resolution ADCs as shown in Fig. \ref{fig:receiver}.
To design a two-stage analog combiner for the proposed architecture by considering the coarse quantization effect, we formulate a MI maximization problem without imposing a constant modulus constraint on an analog combiner. 
We first derive a near optimal analog combining solution for  general channels.
The derived solution can be decomposed into two parts: a channel gain aggregation function that captures channel gains into the lower dimension and a spreading function that evenly spreads the aggregated gains over all available RF chains to reduce quantization error.
We show that the derived solution achieves the optimal scaling law which scales logarithmically with respect to the number of RF chains whereas a conventional optimal solution only achieves a bounded MI.  
The derived solution also maximizes the MI when the singular values of a MIMO channel are the same.
We further propose an ARV-based two-stage analog combining algorithm for mmWave channels to implement the derived solution under the constant modulus constraint. 
Simulation results demonstrate that the proposed two-stage analog combining algorithm outperforms conventional algorithms.

{\it Notation}: $\bf{A}$ is a matrix and $\bf{a}$ is a column vector. 
$\mathbf{A}^{H}$ and $\mathbf{A}^T$  denote conjugate transpose and transpose. 
$[{\bf A}]_{i,:}$ and $ \mathbf{a}_i$ indicate the $i$th row and column vector of $\bf A$. 
We denote $a_{i,j}$ or $[\bA]_{i,j}$ as the $\{i,j\}$th element of $\bf A$ and $a_{i}$ as the $i$th element of $\bf a$. $\lambda_i\{{\bf A}\}$ denotes the $i$th largest singular value of ${\bf A}$. 
$\mathcal{CN}(\mu, \sigma^2)$ is the complex Gaussian distribution with mean $\mu$ and variance $\sigma^2$. 
$\mathbb{E}[\cdot]$ and $\bbV[\cdot]$ represent an expectation and variance operators, respectively.
The correlation matrix is denoted as ${\bf R}_{\bf xy} = \mathbb{E}[{\bf x}{\bf y}^H]$.
The diagonal matrix $\rm diag\{\bf A\}$ has $\{a_{i,i}\}$ at its $i$th diagonal entry, and $\rm diag \{\bf a\}$ or ${\rm diag}\{{\bf a}^T\}$ has $\{a_i\}$ at its $i$th diagonal entry. 
${\rm blkdiag}\{\bA_1,\cdots,\bA_N\}$ is a block diagonal matrix with diagonal entries $\bA_1,\cdots,\bA_N$.
${\bf I}$ denotes the identity matrix with a proper dimension and we indicate the dimension $N$ by $\bI_N$ if necessary.
$\bf 0$ denotes a matrix that has all zeros in its elements with a proper dimension.
$\|\bf A\|$ represents $L_2$ norm. 
$|\cdot|$ indicates an absolute value, cardinality, and determinant for a scalar value $a$, a set $\cA$, and a matrix $\bA$, respectively.
${\rm Tr}\{\cdot\}$ is a trace operator and  $x(N)\sim y(N)$ indicates $\lim_{N\to\infty}\frac{x}{y}=1$.


\section{System Model}
\label{sec:system}

We consider a single-cell uplink network.
A base station (BS) is equipped with $N_r$ receive antennas in uniform linear arrays (ULA) and $N_{\rm RF}$ RF chains ($N_{\rm RF} < N_r$), and serves $N_u$ users each with a single transmit antenna ($N_u \leq N_{\rm RF}$).
We consider that each channel ${\bf h}_{\gamma,k}$ is the sum of $L_k$ propagation paths for user $k$ \cite{ertel1998overview}.
The number of channel paths $L_k$ is considered to be small due to the sparse nature of mmWave channels \cite{rappaport2013millimeter}. 
The narrowband channel of user $k$ is given as
\begin{align}
	\label{eq:channel_geo}
	{\bf h}_{\gamma,k} = \frac{1}{\sqrt{\gamma_k}} \bh_{k} = \sqrt{\frac{N_r}{\gamma_k L_k}}\sum_{\ell = 1}^{L_k}g_{\ell,k} {\bf a}(\phi_{\ell,k})
\end{align}
where $\gamma_k$, $g_{{\ell},k}$, and ${\bf a}(\phi_{{\ell},k})$ are the pathloss, the complex gain of the $\ell${th} propagation path of user $k$, and the ARV for the azimuth AoA of the $\ell$th path of the $k$th user $\phi_{{\ell},k} \in [-\pi/2,\pi/2]$, respectively. 
We assume that $g_{{\ell},k}$ follows an independent and identically distributed (i.i.d.) complex Gaussian distribution, $g_{{\ell},k} \overset{i.i.d}{\sim} \mathcal{CN}(0, 1)$.
The ARV ${\bf a}(\theta)$ for the ULA is given as ${\bf a}(\theta) = \frac{1}{\sqrt{N_r}}\Big[1,e^{-j \pi\vartheta},e^{-j 2\pi \vartheta},\dots,e^{-j (N_r-1)\pi\vartheta}\Big]^T$
where $\vartheta = \frac{2d}{\lambda}\sin(\theta)$ denotes the spatial angle that is related to the physical AoA $\theta$, $d$ represents the distance between antennas, and $\lambda$ is the wave length.
In this paper, $\phi$ and  $\theta$ denote the physical AoAs of a user channel and physical angles of analog combiners, respectively, and $\varphi$ and $\vartheta$ indicate the spatial angles for $\phi$ and $\theta$ where $ \varphi, \vartheta \in [-1,1]$, respectively.


We consider a homogeneous long-term received SNR network\footnote{We remark that the derived results in this paper can also be valid for a heterogeneous long-term received SNR network with minor modification.} for simplicity where the same long-term received SNR for all users is achieved by using a conventional uplink power control that compensates for the large scale fadings \cite{simonsson2008uplink, tejaswi2013survey}. 
Let $\bx =\bP\bs$ be the $N_u \times 1$ vector of transmit signals of $N_u$ users where $\bP ={\rm diag}\{\sqrt{\rho\, \gamma_1},\dots,\sqrt{\rho\,\gamma_{N_u}}\}$ is the matrix of transmit power and $\bs$ is the $N_u \times 1$ transmitted symbol vector.
Let $\bH_\gamma = [\bh_{\gamma,1}, \dots,  \bh_{\gamma,N_u} ]=\bH\bB$ where $\bB = {\rm diag}\{\sqrt{1/\gamma_1},\dots,\sqrt{1/\gamma_{N_u}}\}$.
The received analog baseband signals are given as
\begin{align}
    \nonumber
	\br = \bH_\gamma \bx + \bn =  \sqrt{\rho}\bH\bs + \bn
\end{align} 
where $\bn \sim \cC\cN({\bf 0},\bI_{N_r})$ indicates the additive white Gaussian noise vector.
We also assume zero mean and unit variance for the user symbols $\bs$.
Here, we consider $\rho$ to be the SNR due to the unit variance of the noise. 
The received signal $\br$ is combined via two analog combiners, and we have
\begin{align}
	\nonumber
	\by &= \sqrt{\rho}\bW_{\rm RF_2}^H\bW_{\rm RF_1}^H \bH \bs + \bW_{\rm RF_2}^H\bW_{\rm RF_1}^H\bn  \\ 
	\label{eq:y} 
	& = \sqrt{\rho}\bW_{\rm RF}^H \bH \bs + \bW_{\rm RF}^H\bn
\end{align}
where $\bW_{\rm RF}\! =\! \bW_{\rm RF_1}\bW_{\rm RF_2}$ is the two-stage analog combiner.

Each real and imaginary part of the combined signal $\by$ in \eqref{eq:y} are quantized at ADCs with $b$ quantization bits.
Under the assumptions of a MMSE scalar quantizer, we adopt an additive quantization noise model (AQNM) \cite{fletcher2007robust} which shows reasonable accuracy in the low to medium SNR ranges \cite{orhan2015low}. 
The AQNM provides the approximated linearization of quantization process, which is equivalent to the approximation through Bussgang decomposition for low-resolution ADCs \cite{mezghani2012capacity}.
The quantized signal vector is expressed as \cite{fletcher2007robust,mezghani2012capacity}
\begin{align}
	\label{eq:yq}
	\by_{\rm q} &= \cQ(\by) 
	= \alpha_b \sqrt{\rho}\bW_{\rm RF}^H \bH \bs + \alpha_b \bW_{\rm RF}^H \bn + \bq
\end{align}
where $\cQ(\cdot)$ is the element-wise quantizer,  $\alpha_b = 1 -\beta_b$ is the quantization gain where $\beta_b = \bbE[|y-y_{\rm q}|^2]/\bbE[|y|^2]$, and $\bq$ denotes the quantization noise vector that is uncorrelated with the quantization input $\by$ \cite{fletcher2007robust}.
For $b > 5$ quantization bits, $\beta_b$ is approximated as $\beta_b \approx \frac
{{\pi}\sqrt{3}}{2}2^{-2b}$ for Gaussian transmit signals $\bs \sim \cC\cN({\bf 0}, \bI_{N_u})$.
For $b \leq 5$, the values of $\beta_b$ are listed in Table 1 in \cite{fan2015uplink}.
Here, we assume $\bq \sim \cC\cN({\bf 0}, \bR_{\bq \bq})$, where the covariance matrix $\bR_{\bq\bq}$ is given as \cite{fletcher2007robust}
\begin{align}
	\label{eq:Rqq}
	\bR_{\bq \bq}\! = \!\alpha_b\beta_b{\rm diag}\big\{\rho\bW_{\rm RF}^H \bH \bH^H\bW_{\rm RF}\! +\! \bW_{\rm RF}^H\bW_{\rm RF}\big\}.
\end{align}
Then, $\by_{\rm q}$ is combined through a digital combiner $\bW_{\rm BB}$.




\section{Two-Stage Analog Combining}
\label{sec:twostage}

\subsection{Optimality of Two-Stage Analog Combining}
\label{subsec:analysis}

In this section, we derive a near optimal structure for the first and second analog combiners $\bW_{\rm RF_1}, \bW_{\rm RF_2}$ in low-resolution ADC systems for a general channel by solving a MI maximization problem without a constant modulus condition on the analog combiner $\bW_{\rm RF}$. 
We consider the MI between $\bs$ and $\by_{\rm q}$ under the AQNM model, and it is given as
\begin{align}
	\label{eq:MI}
	\cC(\bW_{\rm RF})= \log_2 \Big|\bI_{N_{\rm RF}} +\rho\alpha_b^2\bD^{-1}\bW_{\rm RF}^H{\bf H}{\bf H}^H\bW_{\rm RF}\Big|. 
\end{align}
where $\bD = \alpha_b^2 \bW_{\rm RF}^H\bW_{\rm RF} + \bR_{\bq\bq}$. 
Based on \eqref{eq:MI}, we formulate a relaxed MI maximization problem as 
\begin{align}
	\label{eq:P1}
    \cP1: ~ \bW_{\rm RF}^{\rm opt} = \argmax_{\bW_{\rm RF}} ~ \cC(\bW_{\rm RF}),~ \text{s.t. } \bW_{\rm RF}^H\bW_{\rm RF}=\bI.
\end{align}
Note that we only assume a semi-unitary constraint on the analog combiner $\bW_{\rm RF}^H\bW_{\rm RF}=\bI_{N_{\rm RF}}$ as in \cite{mo2017hybrid}. 

We first derive an optimal scaling law 
with respect to $N_{\rm RF}$, and provide a solution that achieves the scaling law.
\begin{theorem}[Optimal scaling law]
\label{thm:optimality_two_stage}
    For fixed $N_{\rm RF}/N_r = \kappa$ with $\kappa \in (0,1)$, the MI with the optimal combiner $\bW_{\rm RF}^{\rm opt}$ for the problem $\cP_1$ scales with $N_{\rm RF}$ as 
    \begin{align}  
       \label{eq:C_opt}
        \cC(\bW_{\rm RF}^{\rm opt}) \sim N_u \log_2 N_{\rm RF}
    \end{align}
    and \eqref{eq:C_opt} is achieved by using $\bW_{\rm RF}^\star=\bW_{\rm RF_1}^\star \bW_{\rm RF_2}^\star$ such that:
    \begin{itemize}
        \item[$(i)$] $\bW_{\rm RF_1}^\star = [\bU_{1:N_u} ~ \bU_\perp]$, and 
        \item[$(ii)$] $\bW_{\rm RF_2}^\star$ is any $N_{\rm RF} \times N_{\rm RF}$ unitary matrix that satisfies the constant modulus condition on its elements,  
    \end{itemize}  
    where $\bU_{1:N_u}$ is the matrix of left singular vectors for the first $N_u$ largest singular values of $\bH$ and $\bU_\perp$ is the matrix of any orthonormal vectors such that ${\rm Span}(\bU_\perp) \perp {\rm Span}(\bU_{1:N_u})$.
\end{theorem}
\begin{proof}
    
We derive an upper bound of $\cC(\bW_{\rm RF})$ and its scaling law with respect to $N_{\rm RF}$, and show that adopting $\bW_{\rm RF}^\star = \bW_{\rm RF_1}^\star \bW_{\rm RF_2}^\star$ in Theorem \ref{thm:optimality_two_stage} achieves the same scaling law of the upper bound. 
Let $\bu_i$ be the $i$th left singular vector of $\bH$.
Then, an arbitrary semi-unitary $\bW_{\rm RF}$ can be decomposed into 
\begin{align}\label{eq:WRF_decomposed}
   \bW_{\rm RF} = [\bU_{||}~\bU_{\perp}] \bar{\bW}_{\rm RF},
\end{align}
where $\bU_{||}$ is an $N_r \!\times\! m$ matrix composed of $m$ orthonormal vectors whose column space is in the subspace of ${\rm Span}(\!\bu_1,\!\cdots\!,\bu_{N_u}\!)$ with $1\!\le \!m\! \le \!N_u$, $\bU_{\perp}$ is an $N_r \!\times \!(N_{\rm RF}-m)$ matrix composed of ($N_{\rm RF}\!-\!m$) orthonormal vectors whose column space is in the subspace of ${\rm Span}^\perp(\bu_1,\cdots,\bu_{N_u})$, and $\bar{\bW}_{\rm RF}$ is an $N_{\rm RF} \times N_{\rm RF}$ unitary matrix.

Using \eqref{eq:WRF_decomposed}, $\bW_{\rm RF}^H \bH\bH^H\bW_{\rm RF}$ in \eqref{eq:MI} can be re-written as
    \begin{align}\nonumber
       &\bW_{\rm RF}^H \bH\bH^H\bW_{\rm RF} \\
       \nonumber
       &= \bar{\bW}_{\rm RF}^H [\bU_{||}~\bU_{\perp}]^H \bU {\pmb \Lambda} \bU^H [\bU_{||}~\bU_{\perp}]\bar{\bW}_{\rm RF}\\ 
       &=\bar{\bW}_{\rm RF}^H \underbrace{\left[ 
       \begin{matrix}
       {\bU_{||}^H\bU_{1:N_u}\pmb \Lambda_{N_u} \bU_{1:N_u}^H\bU_{||}} & {\bf 0} \\
       {\bf 0} & {\bf 0}
       \end{matrix}\right]}_{\triangleq \bQ} \bar{\bW}_{\rm RF} \label{eq:wHHw_change}
    \end{align}
where ${\pmb \Lambda}={\rm diag}\{\lambda_1,\cdots,\lambda_{N_u},0,\cdots,0\} \in \bbC^{N_r\times N_r}$, 
$\pmb \Lambda_{N_u} ={\rm diag}\{ \lambda_1,\dots,\lambda_{N_u}\}$, $\lambda_i$ denotes $\lambda_i\{\bH\bH^H\}$, and  $\bU_{1:N_r}=[\bu_1,\cdots,\bu_{N_r}]$. 
The matrix $\bQ$ is a rank $m$ matrix and can be represented as $\bQ=\bU_{\bQ}\bar{\pmb \Lambda}\bU_{\bQ}^H$, where $\bU_{\bQ}$ is the $N_{\rm RF} \times N_{\rm RF}$ matrix consisting of $N_{\rm RF}$ singular vectors of $\bQ$; and $\bar{\pmb \Lambda}={\rm diag}\{\bar{\lambda}_1,\cdots,\bar{\lambda}_m,0,\cdots,0\} \in \bbC^{N_{\rm RF} \times N_{\rm RF}}$.
Here, $\bar{\lambda}_i$ indicates $\lambda_i\{\bQ\}$. Since $\bU_{\bQ}$ is unitary, we rewrite $\bar{\bW}_{\rm RF}$ as 
\begin{align}\label{eq:basis_change}
    \bar{\bW}_{\rm RF}=\bU_{\bQ}\overline{\bW}_{\rm RF}.
\end{align}
where $\overline{\bW}_{\rm RF}$ is a unitary matrix. 
Substituting \eqref{eq:basis_change} into \eqref{eq:wHHw_change}, we have $\bW_{\rm RF}^H\bH\bH^H\bW_{\rm RF} = \overline{\bW}_{\rm RF}^H\bar{\pmb \Lambda}\overline{\bW}_{\rm RF}$ and \eqref{eq:MI} becomes
\begin{align} 
	\nonumber
    & \cC(\bW_{\rm RF}) \\ 
	\label{eq:c_wrf_wlambdaw}
     &\!= \!\log_2\!\left|\bI\! +\! \frac{\alpha_b}{\beta_b} {\rm diag}^{-1}\!\!\left\{\overline{\bW}_{\rm RF}^H\bar{\pmb \Lambda}\overline{\bW}_{\rm RF}\!+\!\frac{1}{\beta_b\rho}\bI \right\}\!\overline{\bW}_{\rm RF}^H\bar{\pmb \Lambda}\overline{\bW}_{\rm RF} \right|. 
    \end{align}
Let $\bG = \overline{\bW}_{\rm RF}^H \bar{\pmb \Lambda}^{1/2}=[\bG_{\rm sub}~ {\bf 0}]$, where $\bG_{\rm sub}$ is the $N_{\rm RF} \times m$ submatrix of $\bG$. Then, the MI can be upper bounded as 
\vspace{- 0.5 em}
\begin{align} \nonumber
    &\cC(\bW_{\rm RF}) \!= \!\log_2\left|\bI_{N_{\rm RF}}\! +\! \frac{\alpha_b}{\beta_b} \bG^H{\rm diag}^{-1}\!\left\{\|[\bG]_{i,:}\|^2+\frac{1}{\beta_b\rho}\right\}\!\bG \right| \\ \nonumber
    &= \log_2\left|\bI_m + \frac{\alpha_b}{\beta_b} \bG_{\rm sub}^H{\rm diag}^{-1}\left\{\|[\bG_{\rm sub}]_{i,:}\|^2+\frac{1}{\beta_b\rho}\right\}\bG_{\rm sub} \right| \\ \nonumber
    &\overset{(a)}{=}\log_2\left|\bI_m + \frac{\alpha_b}{\beta_b} \tilde{\bG}_{\rm sub}^H\tilde{\bG}_{\rm sub} \right| \\ \nonumber
    &=\sum_{i=1}^m \log_2 \left( 1+\frac{\alpha_b}{\beta_b}\lambda_i\{\tilde{\bG}_{\rm sub}^H\tilde{\bG}_{\rm sub}\} \right) \\ \nonumber
    &\overset{(b)}{\le} m \log_2\left(1+\frac{\alpha_b}{\beta_b m}\sum_{i=1}^m \lambda_i\{\tilde{\bG}_{\rm sub}^H\tilde{\bG}_{\rm sub}\}\right) \\
    &\overset{(c)}{=}m \log_2\left(1+\frac{\alpha_b}{\beta_b m}\sum_{i=1}^{N_{\rm RF}} \frac{\|[\bG_{\rm sub}]_{i,:}\|^2}{\|[\bG_{\rm sub}]_{i,:}\|^2+\frac{1}{\beta_b\rho}} \right)  
    \label{eq:upper_MI_bound}
\end{align}
where $(a)$ comes from letting $\tilde{\bG}_{\rm sub}$ be the matrix whose each row $i$ is given as  $i$th row of $\bG_{\rm sub}$ normalized by $\big(\|[\bG_{\rm sub}]_{i,:}\|^2+\frac{1}{\beta_b\rho} \big)^{1/2}$; $(b)$ follows from Jensen's inequality and the concavity of $\log(1+x)$ for $x>0$; and $(c)$ is from 
\begin{align} \nonumber
	\sum_{i=1}^m \! \lambda_i\{\tilde{\bG}_{\rm sub}^H\tilde{\bG}_{\rm sub}\} \!=\! {\rm Tr}\{\tilde{\bG}_{\rm sub}^H\tilde{\bG}_{\rm sub}\}\!=\!\sum_{i=1}^{N_{\rm RF}}\! \frac{\|[\bG_{\rm sub}]_{i,:}\|^2}{\|[\bG_{\rm sub}]_{i,:}\|^2\!+\!\frac{1}{\beta_b\rho}}.
\end{align}
Then, \eqref{eq:upper_MI_bound} is further upper bounded by $ m\log_2(1+\frac{\alpha_b N_{\rm RF}}{\beta_b m})$ because $\frac{\|[\bG_{\rm sub}]_{i,:}\|^2}{\|[\bG_{\rm sub}]_{i,:}\|^2+\frac{1}{\beta_b\rho}} < 1$.
Since $m\log_2(1+\frac{\alpha_b N_{\rm RF}}{\beta_b m})$ is an increasing function of $m$ for $m>0$, it is maximized with $m=N_{u}$, and scales as $N_u \log_2 N_{\rm RF}$ with $N_{\rm RF} \to \infty$. 

Now, we prove that the scaling law can be achieved by using $\bW_{\rm RF}^\star\! =\! \bW_{\rm RF_1}^\star\bW_{\rm RF_2}^\star$.
Let $ \bC\! \triangleq\! \bW_{\rm RF_2}^{\star H}{\pmb \Lambda}_{N_{\rm RF}}\!\bW_{\rm RF_2}^{\star}$.
From $\bW_{\rm RF}^{\star H}\bH\bH^H\bW_{\rm RF}^{\star}\!\! =\!\! \bW_{\rm RF_2}^{\star H}{\pmb \Lambda}_{N_{\rm RF}}\!\!\bW_{\rm RF_2}^{\star}\!\! =\!\! \bC$
 where ${\pmb \Lambda}_{N_{\rm RF}}\!=\! {\rm diag}\{\lambda_1,\cdots,\lambda_{N_u},0,\cdots,0\} \in \bbC^{N_{\rm RF}\times N_{\rm RF}}$ and \eqref{eq:c_wrf_wlambdaw}, we have 
\begin{align}
	\label{eq:c_w1w2}
    &\cC(\bW_{\rm RF}^{\star})\!=\! \log_2\left|\bI_{N_{\rm RF}} \!+\! \frac{\alpha_b}{\beta_b} {\rm diag}^{-1}\!\left\{\bC+\tfrac{1}{\beta_b\rho}\bI_{N_{\rm RF}}\right\}\!\bC \right| \\
    \label{eq:W2_spread}
    &\overset{(a)}{=}\!\log_2\left|\bI \!+\! \frac{\alpha_b}{\beta_b}\!\left(\frac{\sum_{i=1}^{N_u}\lambda_i}{N_{\rm RF}}\!+\!\frac{1}{\beta_b\rho}\right)^{\!\!-1}\!\!\!\!\!\bW_{\rm RF_2}^{\star H}{\pmb \Lambda}_{N_{\rm RF}}\bW_{\rm RF_2}^{\star} \right| \\ 
    \label{eq:achievable_rate_DFT_U}
    &= \sum_{k=1}^{N_u} \log_2\left(1+\frac{\alpha_b\rho N_{\rm RF} \lambda_k/N_r}{\kappa+(1-\alpha_b)\rho \sum_{i=1}^{N_u}\lambda_i/N_r} \right) \\ 
    \nonumber
    &\overset{(b)}{\sim} N_u \log_2 N_{\rm RF},~\text{as}~N_{\rm RF}\to\infty.
\end{align}
where $\kappa = N_{\rm RF}/N_r$. 
Here, $\!(a)$ follows from the fact that all diagonal entries of $\bW_{\rm RF_2}^{\star H}{\pmb \Lambda}_{N_{\rm RF}}\!\bW_{\rm RF_2}^{\star}$ are equal to each other as $d_j={\sum_{i=1}^{N_u}\lambda_i}/{N_{\rm RF}}$, $\forall j$ due to the constant modulus property of $\bW_{\rm RF_2}^{\star}$; $(b)$ is from the fact that as $N_{\rm RF}\to\infty$ ($N_r \to \infty$), we have  $\frac{1}{N_r}\bH^H\bH \to {\rm diag}\{\frac{1}{L_1}\sum_{\ell=1}^{L_1}|g_{\ell,1}|^2,\cdots, \frac{1}{L_{N_u}}\sum_{\ell=1}^{L_{N_u}}|g_{\ell,N_u}|^2\}$ \cite{ngo2014aspects} from the channel $\bh_k$ \eqref{eq:channel_geo} and the law of large numbers, which implies
\begin{align} \nonumber
    \frac{\lambda_i}{N_r} \to  \frac{1}{L_i}\sum_{\ell=1}^{L_i}|g_{\ell, i}|^2< \infty,~\text{for}~i=1,\cdots,N_u.
\end{align}
This completes the proof of Theorem \ref{thm:optimality_two_stage}. 
\end{proof}

We note that $\bW_{\rm RF_1}^\star$ in Theorem \ref{thm:optimality_two_stage} aggregates all channel gains into the smaller dimension and provides ($N_{\rm RF} - N_u$) extra dimensions.
As observed in \eqref{eq:W2_spread}, $\bW_{\rm RF_2}^\star$, then, spreads the aggregated channels gains over all $N_{\rm RF}$ dimensions, thereby reducing the quantization error by exploiting the extra dimensions.
Accordingly, the proposed solution $\bW_{\rm RF}^\star=\bW_{\rm RF_1}^\star\bW_{\rm RF_2}^\star$ achieves the optimal scaling law \eqref{eq:C_opt} by reducing the quantization error as $N_{\rm RF}$ increases. 
We note that $\bW_{\rm RF_1}^\star$ itself is the conventional optimal analog combiner for the problem $\cP1$ without quantization error, i.e., for a perfect quantization system with infinite quantization resolution. 
\begin{corollary}\label{cor:conventional_sol}
      The conventional optimal solution $\bW_{\rm RF}^{\rm cv}= [\bU_{1:N_{u}} ~ \bU_\perp]$ for perfect quantization systems cannot achieve the optimal scaling law \eqref{eq:C_opt} in coarse quantization systems, and the achievable MI with $\bW_{RF}^{\rm cv}$ is upper bounded by
    \begin{align}
   	 	\label{eq:bounded_performance}
        \cC\big(\bW_{\rm RF}^{\rm cv}\big)  <  \cC_{\rm svd}^{\rm ub} = N_u \log_2\big(1+{\alpha_b}/{(1-\alpha_b)}\big). 
    \end{align} 
\end{corollary}
\begin{proof}
     From \eqref{eq:c_w1w2}, we have the following MI with $\bW_{\rm RF_2} = \bI$: 
     \begin{align}
        \nonumber
         &\cC\big(\bW_{\rm RF}^{\rm cv}\big)= \log_2\left|\bI+ \frac{\alpha_b}{\beta_b} {\rm diag}^{-1}\left\{{\pmb \Lambda}_{N_{\rm RF}}+\tfrac{1}{\beta_b\rho}\bI\right\}{\pmb \Lambda}_{N_{\rm RF}} \right|\\
         \nonumber
         &= \sum_{i=1}^{N_u} \log_2\left(1+\frac{\alpha_b\lambda_i}{\beta_b\lambda_i + {1}/{\rho}}\right)\stackrel{(a)}< N_u \log_2\left(1+\frac{\alpha_b}{\beta_b}\right).
    \end{align}
    where $(a)$ comes from $\rho > 0 $.
\end{proof} 
Corollary \ref{cor:conventional_sol} shows that although the conventional optimal combiner $\bW_{\rm RF}^{\rm cv}$ captures the entire channel gains, the MI does not scale as the MI with  $\bW_{\rm RF}^\star=\bW_{\rm RF_1}^\star\bW_{\rm RF_2}^\star$.
Since the channel gains after $\bW_{\rm RF}^{\rm cv}$ are concentrated on only $N_u$ RF chains, this results in severe quantization errors at each of the $N_u$ RF chains. 
Therefore, while the channel gains $\{\lambda_i\}$ increase with $N_r$, the quantization errors also increase, leading to the bounded MI in \eqref{eq:bounded_performance}. 
This confirms the benefit of using the proposed second analog combiner $\bW_{\rm RF_2}^\star$. 
In addition, we show the optimality of the proposed two-stage combiner in maximizing the MI for a special case.


\begin{theorem}\label{thm:equal_eigenvalue_optimal}
    When all singular values $\lambda_i$ of $\bH^H\bH$ are equal, the two-stage analog combining solution  $\bW_{\rm RF}^\star =\bW_{\rm RF_1}^\star \bW_{\rm RF_2}^\star$ in Theorem~\ref{thm:optimality_two_stage} maximizes the MI  in \eqref{eq:P1} with finite $N_{\rm RF}$.
    The corresponding optimal MI is given as
    \begin{align}\label{eq:the_optimal_rate_AQNM}
		 \cC_{\rm opt} \! \triangleq \!\cC(\bW_{\rm RF}^\star) \! = \!N_u \!\log_2\!\!\left(\!1\!+\!\frac{\alpha_b\lambda N_{\rm RF}}{\lambda N_u\!(1\!-\!\alpha_b) \!+\!{N_{\rm RF}}/{\rho}} \!\right)\!.
    \end{align}
\end{theorem}
\begin{proof}
    Recall $\bG=\overline{\bW}_{\rm RF}^H\bar{\pmb \Lambda}^{1/2}=[\bG_{\rm sub}~{\bf 0}]$ in the proof of Theorem \ref{thm:optimality_two_stage}, where $\bG_{\rm sub}$ is the $N_{\rm RF} \times m$ submatrix and $\bar{\pmb \Lambda}\!=\!{\rm diag}\{\bar{\lambda}_1,\!\cdots\!,\bar{\lambda}_m,0,\!\cdots\!,0\}$ where $\bar \lambda_i\! =\! \lambda_i\{\bQ\}$ and $\bQ$ is defined in \eqref{eq:wHHw_change}. 
    From the assumption of $\lambda_1=\cdots=\lambda_{N_u} = \lambda $,  we have 
    \begin{align}\nonumber
		\max_{\bx\in\mathbb{C}^{N_{\rm RF}}:\|\bx\|=1}\bx^H \bQ \bx &= \max_{\by\in\mathbb{C}^{m}:\|\by\|=1}\lambda \|\bU_{1:N_u}^H\bU_{||}\by\|^2 \\ \nonumber
        &\overset{(a)}{\le} \max_{\by\in\mathbb{C}^{m}:\|\by\|=1}\lambda \|\bU_{1:N_u}^H\|^2\|\bU_{||}\|^2\|\by\|^2 \\ \nonumber
        &=\lambda,
    \end{align}
    where 
	$(a)$ comes from the sub-multiplicativity of the norm, and the last equality holds by $\|\bU_{1:N_u}^H\|=1$ and $\|\bU_{||}\|=1$. 
	This implies  $\bar{\lambda}_i \le \lambda$ for $i=1,\cdots,m$. 
	Therefore, 
    $\|[\bG_{\rm sub}]_{j,:}\|^2$ is maximized for any $\overline{\bW}_{\rm RF}$ when $\bar{\lambda}_i = \lambda$ for $i=1,\cdots,m$. 
    
    We consider the upper bound of $\cC(\bW_{\rm RF})$ in \eqref{eq:upper_MI_bound} and define  $  \bG^\star_{\rm sub}\!=\!\overline{\bW}_{\rm RF}^H\!\!
        \begin{bmatrix}\!
        \sqrt{\lambda} \bI_m \\ 
        {\bf 0}
        \!\end{bmatrix}\!\!\in\! \bbC^{N_{\rm RF} \!\times\! m}$.
    Then, the upper bound of $\cC(\bW_{\rm RF})$ in \eqref{eq:upper_MI_bound} is further upper bounded by
    \begin{align} \nonumber
        \cC(\bW_{\rm RF}) &\le m \log_2\left(1+\frac{\alpha_b}{\beta_b m}\sum_{i=1}^{N_{\rm RF}} \frac{\|[\bG^\star_{\rm sub}]_{i,:}\|^2}{\|[\bG^\star_{\rm sub}]_{i,:}\|^2+\frac{1}{\beta_b\rho}} \right) \\ \nonumber
        &\overset{(a)}{\le} m \log_2\left(\!1+ \frac{\alpha_b N_{\rm RF}\sum_{i=1}^{N_{\rm RF}}\|[\bG^\star_{\rm sub}]_{i,:}\|^2}{\beta_b m\left(\sum_{i=1}^{N_{\rm RF}}\|[\bG^\star_{\rm sub}]_{i,:}\|^2+\frac{N_{\rm RF}}{\beta_b\rho}\right)} \! \right) \\
        &\overset{(b)}{=}m\log_2\left(1+ \frac{\alpha_b \lambda N_{\rm RF} }{\lambda m {\beta_b}+ {N_{\rm RF}}/{\rho}} \right), \label{eq:upper_MI_equal_eigen}
    \end{align}
    where $(a)$ is from Jensen's inequality and the concavity of ${x}/{(x+1)}$ for $x>0$; and $(b)$ is from $ \sum_{i=1}^{N_{\rm RF}}\|[\bG^\star_{\rm sub}]_{i,:}\|^2 = \|\bG^\star_{\rm sub}\|_F^2 = \lambda m.$
     Note that \eqref{eq:upper_MI_equal_eigen} is maximized with $m=N_u$ as \eqref{eq:upper_MI_equal_eigen} is an increasing function of $m$ for $m>0$ with $\alpha_b,\lambda,\rho,N_{\rm RF}>0$. 
     By putting $\lambda_1=\cdots=\lambda_{N_u} = \lambda$ into \eqref{eq:achievable_rate_DFT_U}, it is shown that the upper bound of $\cC(\bW_{\rm RF})$ in \eqref{eq:upper_MI_equal_eigen} with $m=N_u$ can be achieved by adopting $\bW_{\rm RF}^{\star}=\bW_{\rm RF_1}^{\star}\bW_{\rm RF_2}^{\star}$.
\end{proof}

\begin{algorithm}[t!]
\label{algo:ARV}
 {\bf Initialization}: set $\bW_{\rm RF_1} = $ empty matrix, $\bH_{\rm rm}=\bH$, and $\cV = \{\vartheta_1,\dots, \vartheta_{|\mathcal{V}|}\}$ where $\vartheta_n = \frac{2n}{|\mathcal{V}|}-1$ \\
 \For{$i = 1:N_{\rm RF}$}{
    Maximum channel gain aggregation
 \begin{enumerate}
     \item[(a)] $\ba(\vartheta^\star) = \argmax_{\vartheta \in \mathcal{V}} \|\ba(\vartheta)^H\bH_{\rm rm}\|^2$
     \vspace{0.1em}
     \item[(b)] $\bW_{\rm RF_1} = \big[\ \bW_{\rm RF_1}\ |\  \ba(\vartheta^\star)\ \big] $
     \vspace{0.1em}
     \item[(c)] $\bH_{\rm rm}=\cP_{\ba(\vartheta^\star)}^{\perp}\bH_{\rm rm}$, where $\cP_{\ba(\vartheta)}^{\perp}\!=\!\bI\! -\! \ba(\vartheta)\ba(\vartheta)^H$
     \vspace{0.1em}
     \item[(d)]$\cV = \cV\setminus\{\vartheta^\star\}$
     \vspace{0.1em}
 \end{enumerate}
 }
Set $\bW_{\rm RF_2} = \bW_{\rm DFT}$ where $\bW_{\rm DFT}$ is a normalized $N_{\rm RF} \times N_{\rm RF}$ DFT matrix.\\
\Return{\ }{$\bW_{\rm RF_1}$ and $\bW_{\rm RF_2}$\;}
\caption{ARV-based TSAC}
\end{algorithm}

\subsection{Two-Stage Analog Combining Algorithm}


We propose an ARV-based two-stage analog combining (ARV-TSAC) algorithm for mmWave channels to implement the derived combining solution in Theorem \ref{thm:optimality_two_stage} under the constant modulus constraint.
Theorem \ref{thm:optimality_two_stage} provides a practical analog combiner structure that is implementable with a two-stage analog combiner $\bW_{\rm RF} = \bW_{\rm RF_1}\bW_{\rm RF_2}$ under the constant modulus constraint.
Since the ARV can be considered as a basis of mmWave channels as shown in \eqref{eq:channel_geo}, finding a set of ARVs that are orthogonal to each other and collects most channel gains can perform similar to using $[\bU_{1:N_u} \bU_\perp]$ for $\bW_{\rm RF_1}$.
In this regard,
we adopt an ARV-codebook based maximum channel gain aggregation approach  
to capture most channel gains into the fewer RF chains by exploiting the sparse nature of mmWave channels. 
To this end, we first set the codebook of the evenly spaced spatial angles $\cV = \{\vartheta_1,\dots,\vartheta_{|\mathcal{V}|}\}$.
To avoid excessive search complexity for the exhaustive method, the proposed algorithm operates in greedy manner to find the best $N_{\rm RF}$ ARVs with greatly reduced complexity.

The proposed ARV-TSAC method is described in Algorithm~\ref{algo:ARV}. 
The ARV $\ba(\vartheta^\star)$ which captures the largest channel gain in the remaining channel dimensions $\bH_{\rm rm}$ is selected in Step (a) and composes a column of $\bW_{\rm RF_1}$ in Step (b). 
In Step (c), $\bH_{\rm rm}$ is projected onto the subspace of ${\rm Span}^\perp(\ba(\vartheta^\star))$ to remove the channel gain on the space of $\ba(\vartheta^\star)$.
Algorithm~\ref{algo:ARV} repeats these steps until $N_{\rm RF}$ ARVs are selected from the codebook $\cV$.
The algorithm uses a fixed DFT matrix $\bW_{\rm DFT}$ to implement $\bW_{{\rm RF}_2}^\star$ as $\bW_{\rm DFT}$ satisfies both the unitary and constant modulus constraints.

\begin{figure}[!t]\centering
	\includegraphics[scale = 0.33]{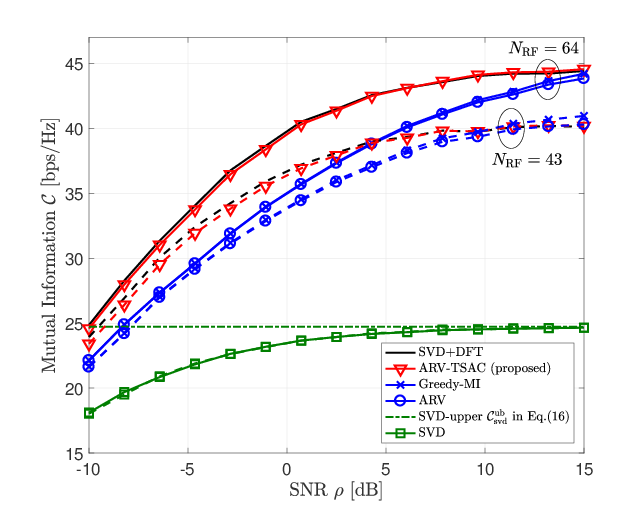}
	\caption{The MI simulation results for $N_r = 128$ receive antennas,  $N_u = 8$ users, $\lambda_L = 3$ average channel paths, $b=2$ quantization bits, and $N_{\rm RF} \in \{43, 64\}$ RF chains that are  $\lceil N_r/3 \rceil$ and $\lceil N_r/2\rceil$, respectively. }
	\label{fig:MI_SNR}
	\vspace{-1em}
\end{figure}
Employing the fixed DFT matrix for the second analog combiner $\bW_{\rm RF_2}=\bW_{\rm DFT}$ provide benefits in reducing implementation complexity and power consumption because $\bW_{\rm DFT}$ does not depend on the channel and can be constructed by using passive (fixed) phase shifters.
Therefore, $\bW_{\rm RF_2}$ can be implemented with very low complexity and power consumption in the practical system.
In addition, when $N_{\rm RF}$ is a power of two, the fast Fourier transform version of the DFT calculation can be applied, which reduces the number of passive phase shifters for $\bW_{\rm RF_2}$ to $N_{\rm RF}\log_2 N_{\rm RF}$.
\section{Simulation Results}
\label{sec:sim}

In this section, we evaluate the performance of the proposed two-stage analog combing algorithm in the MI and ergodic sum rate. 
In the simulations, we set the codebook size to be $|\cV|=N_r$, which guarantees $\bW_{\rm RF}^H\bW_{\rm RF} = \bI_{N_{\rm RF}}$.  
In the simulations, we evaluate the following cases:
\begin{enumerate}
    \item ARV-TSAC: proposed two-stage analog combining.
    \item ARV: analog combining only with $\bW_{\rm RF} = \bW_{\rm RF_1}$ selected from the ARV-TSAC.
    \item SVD+DFT: two-stage analog combining with $\bW_{\rm RF_1} = \bU_{1:N_{\rm RF}}$ and $\bW_{\rm RF_2} = \bW_{\rm DFT}$ based on Theorem \ref{thm:optimality_two_stage}.
    \item SVD: one-stage analog combining $\bW_{\rm RF} = \bU_{1:N_{\rm RF}}$.
    \item Greedy-MI: one-stage analog combining with greedy-based MI maximization.
\end{enumerate}
Note that the SVD+DFT and SVD cases are infeasible to implement in practice due to the constant modulus constraint.
Here, the greedy-MI maximization method is also evaluated to provide a reference performance.
At each iteration, the greedy method searches for a single ARV from the codebook $\cV$ which maximizes the MI with the previously selected ARVs, and repeats until selects $N_{\rm RF}$ ARVs.
For mmWave channels, we adopt $L_k = {\rm max}\{1, {\rm Poisson}(\lambda_L)\}$ \cite{akdeniz2014millimeter} unless mentioned otherwise, where $\lambda_L$ is the average number of channel paths.


\subsection{Mutual Information}

\begin{figure}[t]
\centering
$\begin{array}{c c}
{\resizebox{0.9\columnwidth}{!}
{\includegraphics{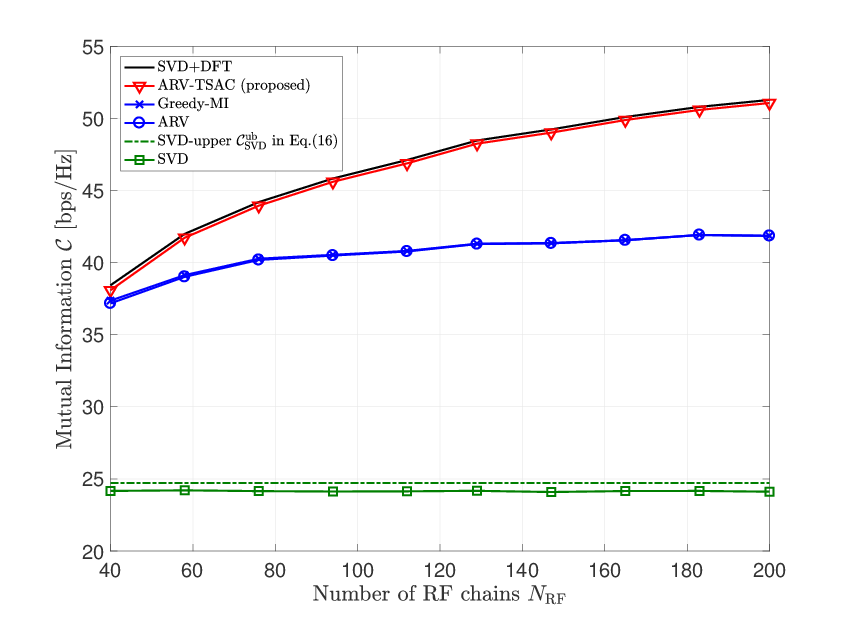}}
}\\ \mbox{\small(a) $N_r = 256$}\\
{\resizebox{0.9\columnwidth}{!}
{\includegraphics{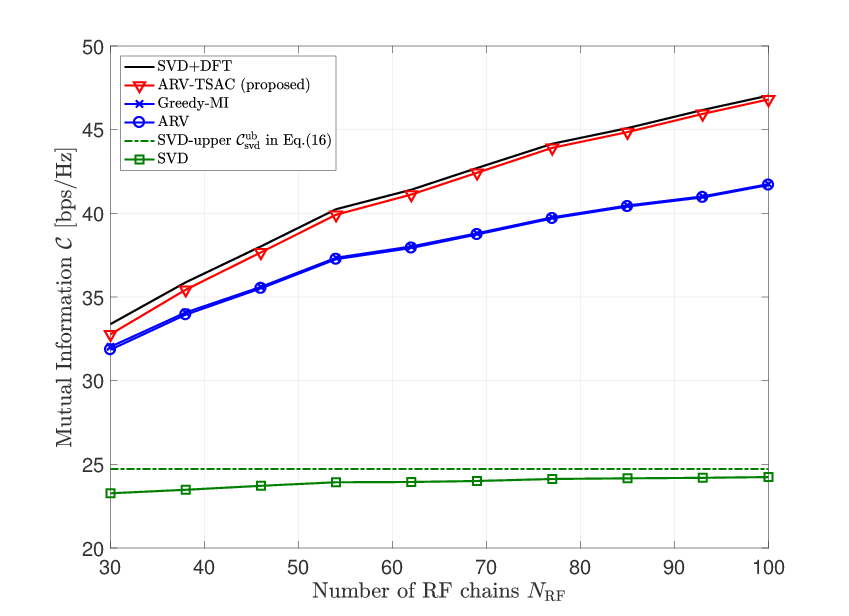}}
}\\\mbox{\small(b) $\kappa = 1/3$}
\end{array}$
\caption{The MI simulation results with $N_u = 8$ users, $\lambda_L = 4$ average channel paths, $b=2$ quantization bits, and $\rho = 0$ dB SNR for (a) $N_r = 256$ receive antennas and (b) $\kappa = N_{\rm RF}/N_r = 1/3$.}
\label{fig:MI_NRF}
\vspace{ -1 em}
\end{figure}

Fig. \ref{fig:MI_SNR} shows the MI simulation results for $N_r = 128$, $N_{\rm RF} \in \{43, 64\}$, $N_u = 8$, $\lambda_L = 3$, and $b=2$ with respect to the SNR $\rho$.
The proposed ARV-TSAC algorithm shows a similar MI as does the SVD+DFT case, and they achieve the highest MI over the most SNR ranges.
Having a gap from the ARV-TSAC, the Greedy-MI and ARV cases achieve similar MI to each other. 
Note that the MI gap decreases as $\rho$ increases in the high SNR regime, and the Greedy-MI and ARV cases with $N_{\rm RF} = 43$ attain the higher MI than that of the SVD+DFT and ARV-TSAC in the very high SNR regime.
This is because the simulated channel environment does not meet the optimality condition for the two-stage analog combining solution.
As more RF chains are used, however, the MI gap becomes larger and the performance reversal would happen in even the higher SNR regime, which corresponds to the intuition derived in Section \ref{sec:twostage}.
The SVD case results in the worst MI performance and its MI converges to the theoretic upper bound $\cC_{\rm svd}^{\rm ub}$ in \eqref{eq:bounded_performance}.

The MI simulation results are shown in Fig. \ref{fig:MI_NRF}  with $N_u = 8$, $\lambda_L = 4$, $b=2$, and $\rho = 0$ dB in terms of  $N_{\rm RF}$.
In Fig. \ref{fig:MI_NRF}(a), the MIs of SVD+DFT and ARV-TSAC   increase logarithmically with fixed $N_r = 256$, which corresponds to the scaling law in Theorem \ref{thm:optimality_two_stage}.
The Greedy-MI, ARV, and SVD cases, however, show only marginal increase of the MI as $N_{\rm RF}$ increases. 
In Fig. \ref{fig:MI_NRF}(b), $\kappa = N_{\rm RF}/N_r$ is fixed to be $\kappa = 1/3$.
The Greedy-MI and ARV cases increase more slowly compared to the SVD+DFT and ARV-TSAC cases.
This is because more channel gains are collected as $N_r$ increases for all cases, but the two-stage combining can reduce more quantization error as $N_{\rm RF}$ increases.
Thus, the MI gap between the two-stage and one-stage combining cases increases as $N_{\rm RF}$ increases. 
\section{Conclusion}
\label{sec:con}

In this paper, we derived a near optimal two-stage analog combining solution for an unconstrained MI maximization problem in hybrid MIMO systems with low-resolution ADCs. 
We showed that unlike a conventional optimal solution, the derived solution achieves the optimal scaling law and maximizes the mutual information for a homogeneous channel singular value case.
We further implemented the solution in the proposed two-stage analog combining architecture that decouples the channel gain aggregation and spreading functions in the solution into two cascaded analog combiners.
Simulation results validated the key insights obtained in this paper and demonstrated that the proposed two-stage analog combining algorithm outperforms conventional one-stage algorithms. 
Therefore, considering the low complexity in deploying the second analog combiner, the proposed two-stage analog combining architecture can provide a better performance and power tradeoff than a conventional hybrid architecture for future wireless communications.

\bibliographystyle{IEEEtran}
\bibliography{ICCtwostage.bib}

\begin{thebibliography}{10}
\providecommand{\url}[1]{#1}
\csname url@samestyle\endcsname
\providecommand{\newblock}{\relax}
\providecommand{\bibinfo}[2]{#2}
\providecommand{\BIBentrySTDinterwordspacing}{\spaceskip=0pt\relax}
\providecommand{\BIBentryALTinterwordstretchfactor}{4}
\providecommand{\BIBentryALTinterwordspacing}{\spaceskip=\fontdimen2\font plus
\BIBentryALTinterwordstretchfactor\fontdimen3\font minus
  \fontdimen4\font\relax}
\providecommand{\BIBforeignlanguage}[2]{{%
\expandafter\ifx\csname l@#1\endcsname\relax
\typeout{** WARNING: IEEEtran.bst: No hyphenation pattern has been}%
\typeout{** loaded for the language `#1'. Using the pattern for}%
\typeout{** the default language instead.}%
\else
\language=\csname l@#1\endcsname
\fi
#2}}
\providecommand{\BIBdecl}{\relax}
\BIBdecl

\bibitem{pi2011introduction}
Z.~Pi and F.~Khan, ``{An introduction to millimeter-wave mobile broadband
  systems},'' \emph{IEEE Commun. Mag}, vol.~49, no.~6, pp. 101--107, Jun. 2011.

\bibitem{rappaport2013millimeter}
T.~S. Rappaport, S.~Sun, R.~Mayzus, H.~Zhao, Y.~Azar, K.~Wang, G.~N. Wong,
  J.~K. Schulz, M.~Samimi, and F.~Gutierrez, ``{Millimeter wave mobile
  communications for 5G cellular: It will work!}'' \emph{IEEE Access}, vol.~1,
  pp. 335--349, May 2013.

\bibitem{andrews2014will}
J.~G. Andrews, S.~Buzzi, W.~Choi, S.~V. Hanly, A.~Lozano, A.~C. Soong, and
  J.~C. Zhang, ``{What will 5G be?}'' \emph{IEEE Journal Sel. Areas Commun.},
  vol.~32, no.~6, pp. 1065--1082, Jun. 2014.

\bibitem{el2014spatially}
O.~El~Ayach, S.~Rajagopal, S.~Abu-Surra, Z.~Pi, and R.~W. Heath, ``{Spatially
  sparse precoding in millimeter wave MIMO systems},'' \emph{IEEE Trans.
  Wireless Commun.}, vol.~13, no.~3, pp. 1499--1513, 2014.

\bibitem{alkhateeb2014channel}
A.~Alkhateeb, O.~El~Ayach, G.~Leus, and R.~W. Heath, ``{Channel estimation and
  hybrid precoding for millimeter wave cellular systems},'' \emph{IEEE Journal
  Sel. Topics Signal Process.}, vol.~8, no.~5, pp. 831--846, 2014.

\bibitem{bogale2014beamforming}
T.~E. Bogale and L.~B. Le, ``{Beamforming for multiuser massive MIMO systems:
  Digital versus hybrid analog-digital},'' \emph{IEEE Global Commun. Conf.},
  2014.

\bibitem{rusu2015low}
C.~Rusu, R.~M{\'e}ndez-Rial, N.~Gonz{\'a}lez-Prelcic, and R.~W. Heath, ``{Low
  complexity hybrid sparse precoding and combining in millimeter wave MIMO
  systems},'' in \emph{IEEE Int. Conf. Commun.}, 2015, pp. 1340--1345.

\bibitem{chen2015iterative}
C.-E. Chen, ``{An iterative hybrid transceiver design algorithm for millimeter
  wave MIMO systems},'' \emph{IEEE Wireless Commun. Lett.}, vol.~4, no.~3, pp.
  285--288, 2015.

\bibitem{liang2014low}
L.~Liang, W.~Xu, and X.~Dong, ``{Low-complexity hybrid precoding in massive
  multiuser MIMO systems},'' \emph{IEEE Wireless Commun. Lett.}, vol.~3, no.~6,
  pp. 653--656, 2014.

\bibitem{alkhateeb2015limited}
A.~Alkhateeb, G.~Leus, and R.~W. Heath, ``{Limited feedback hybrid precoding
  for multi-user millimeter wave systems},'' \emph{IEEE Trans. Wireless
  Commun.}, vol.~14, no.~11, pp. 6481--6494, 2015.

\bibitem{mendez2016hybrid}
R.~M{\'e}ndez-Rial, C.~Rusu, N.~Gonz{\'a}lez-Prelcic, A.~Alkhateeb, and R.~W.
  Heath, ``{Hybrid MIMO architectures for millimeter wave communications: Phase
  shifters or switches?}'' \emph{IEEE Access}, vol.~4, pp. 247--267, Jan. 2016.

\bibitem{venkateswaran2010analog}
V.~Venkateswaran and A.-J. van~der Veen, ``{Analog beamforming in MIMO
  communications with phase shift networks and online channel estimation},''
  \emph{IEEE Trans. Signal Process.}, vol.~58, no.~8, pp. 4131--4143, 2010.

\bibitem{choi2017resolution}
J.~Choi, B.~L. Evans, and A.~Gatherer, ``{Resolution-adaptive hybrid MIMO
  architectures for millimeter wave communications},'' \emph{IEEE Trans. Signal
  Process.}, vol.~65, no.~23, pp. 6201--6216, 2017.

\bibitem{choi2018user}
J.~Choi, G.~Lee, and B.~L. Evans, ``{User Scheduling for Millimeter Wave Hybrid
  Beamforming Systems with Low-Resolution ADCs},'' \emph{arXiv preprint
  arXiv:1804.03079}, 2018.

\bibitem{mo2017hybrid}
J.~Mo, A.~Alkhateeb, S.~Abu-Surra, and R.~W. Heath, ``{Hybrid architectures
  with few-bit ADC receivers: Achievable rates and energy-rate tradeoffs},''
  \emph{IEEE Trans. Wireless Commun.}, vol.~16, no.~4, pp. 2274--2287, 2017.

\bibitem{abbas2017millimeter}
W.~B. Abbas, F.~Gomez-Cuba, and M.~Zorzi, ``{Millimeter wave receiver
  efficiency: A comprehensive comparison of beamforming schemes with low
  resolution ADCs},'' \emph{IEEE Trans. Wireless Commun.}, vol.~16, no.~12, pp.
  8131--8146, 2017.

\bibitem{roth2018comparison}
K.~Roth, H.~Pirzadeh, A.~L. Swindlehurst, and J.~A. Nossek, ``{A Comparison of
  Hybrid Beamforming and Digital Beamforming with Low-Resolution ADCs for
  Multiple Users and Imperfect CSI},'' \emph{IEEE Journal Sel. Topics Signal
  Process.}, 2018.

\bibitem{ertel1998overview}
R.~B. Ertel, P.~Cardieri, K.~W. Sowerby, T.~S. Rappaport, and J.~H. Reed,
  ``{Overview of spatial channel models for antenna array communication
  systems},'' \emph{IEEE Personal Commun.}, vol.~5, no.~1, pp. 10--22, 1998.

\bibitem{simonsson2008uplink}
A.~Simonsson and A.~Furuskar, ``{Uplink power control in LTE-overview and
  performance, subtitle: principles and benefits of utilizing rather than
  compensating for SINR variations},'' in \emph{IEEE Veh. Technol. Conf.},
  2008, pp. 1--5.

\bibitem{tejaswi2013survey}
E.~Tejaswi and B.~Suresh, ``{Survey of power control schemes for LTE uplink},''
  \emph{Int. Journal Computer Science and Inform. Technol.}, vol.~10, p.~2,
  2013.

\bibitem{fletcher2007robust}
A.~K. Fletcher, S.~Rangan, V.~K. Goyal, and K.~Ramchandran, ``{Robust
  predictive quantization: Analysis and design via convex optimization},''
  \emph{IEEE Journal Sel. Topics Signal Process.}, vol.~1, no.~4, pp. 618--632,
  2007.

\bibitem{orhan2015low}
O.~Orhan, E.~Erkip, and S.~Rangan, ``{Low power analog-to-digital conversion in
  millimeter wave systems: Impact of resolution and bandwidth on
  performance},'' in \emph{IEEE Inform. Theory and App. Work.}, Feb. 2015, pp.
  191--198.

\bibitem{mezghani2012capacity}
A.~Mezghani and J.~A. Nossek, ``{Capacity lower bound of MIMO channels with
  output quantization and correlated noise},'' in \emph{IEEE Int. Symp. Inform.
  Theory}, 2012.

\bibitem{fan2015uplink}
L.~Fan, S.~Jin, C.-K. Wen, and H.~Zhang, ``{Uplink achievable rate for massive
  MIMO systems with low-resolution ADC},'' \emph{IEEE Commun. Lett.}, vol.~19,
  no.~12, pp. 2186--2189, 2015.

\bibitem{ngo2014aspects}
H.~Q. Ngo, E.~G. Larsson, and T.~L. Marzetta, ``{Aspects of favorable
  propagation in massive MIMO},'' in \emph{European Signal Process. Conf.},
  2014, pp. 76--80.

\bibitem{akdeniz2014millimeter}
M.~R. Akdeniz, Y.~Liu, M.~K. Samimi, S.~Sun, S.~Rangan, T.~S. Rappaport, and
  E.~Erkip, ``{Millimeter wave channel modeling and cellular capacity
  evaluation},'' \emph{IEEE Journal Sel. Areas Commun.}, vol.~32, no.~6, pp.
  1164--1179, 2014.

\end{thebibliography}

\end{document}